\documentclass[11pt,reqno,a4paper]{amsart}
\pdfoutput=1
\newif\ifshort
\shortfalse
\def\keywords{\smallskip\noindent\textsc{Key Words. }}
\newcommand{\detitle}[1]{#1}
\usepackage{amsfonts,amstext,amsmath,amssymb,stmaryrd,mathrsfs,calc,amsthm,mathtools,bold-extra}

\newcommand{\?}[1]{\mathcal{#1}}
\newcommand{\eqdef}{\stackrel{{\text{\tiny def}}}{=}}

\renewcommand{\vec}[1]{\mathsf{\bar{#1}}}

\newcommand{\tup}[1]{\langle #1\rangle}
\newcommand{\enc}[1]{\ulcorner #1\urcorner}
\newcommand{\imp}{\to}
\newcommand{\Rel}{\mathbf{R}_\imp}
\newcommand{\thetabrule}[1]{\ensuremath{\mathsf{#1}}}
\makeatletter
\providecommand{\Hmake@df@tag@@}[1]{}
\newcommand{\tabrulelabelr}[5][\relax]{
  \Hmake@df@tag@@{#5}%
  \Hy@GlobalStepCount\Hy@linkcounter
  \xdef\@currentHref{equation.\the\c@section.\the\Hy@linkcounter}%
  \Hy@raisedlink{\hyper@anchorstart{\@currentHref}\hyper@anchorend}%
  {\renewcommand{\arraystretch}{1}%
  \frac{ \begin{array}{c} #3 \end{array} }%
       { \begin{array}{c} #4 \end{array} }%
  \@bsphack
  \protected@write\@auxout{}%
    {\string\newlabel{#2}{{{\thetabrule{#5}}}{\thepage}{#1}{\@currentHref}{}}}%
  \:(\thetabrule{#5})}
}
\makeatother
\newcommand{\rlabel}[1]{\RightLabel{{\tiny\eqref{#1}}}}
\newcommand{\jdg}{\mathrel{\triangleright}}%
\newcommand{\fprv}{\Vdash}
\usepackage{bussproofs} %
\usepackage{thmtools,thm-restate}
\newtheorem{theorem}{Theorem}[section]
\newtheorem{lemma}[theorem]{Lemma}
\newtheorem{proposition}[theorem]{Proposition}

\newtheorem{fact}[theorem]{Fact}

\theoremstyle{definition}

\theoremstyle{remark}
\newtheorem{claim}[theorem]{Claim}

\makeatletter%
\usepackage{subfig,tikz}
\usetikzlibrary{positioning}
\usetikzlibrary{arrows}
\usetikzlibrary{automata}
\tikzstyle{intermediate}=[draw=gray!90,very thick,fill=gray!20,circle]
\tikzstyle{state}=[intermediate,minimum width=.65cm,inner sep=1pt]
\tikzstyle{rstate}=[state,rectangle,rounded corners=8pt,minimum
      height=.65cm]
\tikzstyle{every node}=[font=\small]
\tikzstyle{every edge}=[draw,->,>=stealth',shorten >=1pt,semithick]
\tikzstyle{accepting}=[accepting by arrow]
\tikzstyle{initial}=[initial by arrow,initial text=]
\makeatletter
\def\@listi{\leftmargin\leftmargini
               \topsep 3\p@ \@plus\p@ \@minus\p@
               \parsep 2\p@ \@plus\p@ \@minus\p@
               \itemsep \parsep}
\makeatother
\usepackage{hyperref}

\providecommand{\urlstyle}[1]{}
\providecommand{\doi}[1]{\href{http://dx.doi.org/#1}{\nolinkurl{doi:#1}}}
\newcommand{\appref}{\autoref}
\usepackage[numbers]{natbib}
\renewcommand{\cite}{\citep}

\begin{document}
\renewcommand{\sectionautorefname}{Section}
\renewcommand{\subsectionautorefname}{Section}
\renewcommand{\subsubsectionautorefname}[1]{\S}
\title{Implicational Relevance Logic
  is \textsc{2-ExpTime}-Complete} \thanks{Work partially supported by
  ANR grant ReacHard 11-BS02-001-01.}  \author[S.~Schmitz]{Sylvain
  Schmitz} \address{ENS Cachan \& INRIA, France}
\email{schmitz@lsv.ens-cachan.fr}
\begin{abstract}
  We show that provability in the implicational fragment of relevance
logic is complete for doubly exponential time, using reductions to and
from coverability in branching vector addition systems.

\end{abstract}
\maketitle
\section{Introduction}
Relevance logic $\mathbf{R}$~\citep{entailment1,dunn02}
provides a %
formalisation of `relevant'
implication: in such a system, the formula $A\imp B$ indicates that the
truth of $A$ is actually useful in establishing $B$; an example of an
\emph{irrelevant} implication valid in classical logic would be $B\imp
(A\imp B)$.

The pure implicational fragment $\Rel$ of $\mathbf{R}$ was developed
independently by \citet{moh50} in \citeyear{moh50} and
\citet{church51} in \citeyear{church51}, and is as such the oldest of
the relevance logics.  \Citeauthor{kripke59} already presented in
\citeyear{kripke59} a decision algorithm for provability in
$\Rel$~\citep{kripke59}, which was later extended to larger and larger
subsets of $\mathbf{R}$, like the conjunctive-implicational fragment
$\mathbf{R}_{{\imp},{\wedge}}$.  Several negative results by
\citeauthor{urquhart84} would however foil any hope for elementary
algorithms: first in \citeyear{urquhart84} when he showed the
undecidability of the full logic $\mathbf{R}$~\citep{urquhart84};
later in \citeyear{urquhart99} with a proof that
$\mathbf{R}_{{\imp},{\wedge}}$ suffers from a non primitive-recursive
complexity: it is \textsc{Ackermann}-complete~\citep{urquhart99}.
This left a gigantic gap for the implicational fragment $\Rel$,
between an earlier \textsc{ExpSpace} lower bound~\citep{urquhart90}
and the \textsc{Ackermann} upper bound shared by the variants of
\citeauthor{kripke59}'s procedure.

In this paper, we close this gap and show that provability in $\Rel$ is
\mbox{\textsc{2-ExpTime}}-complete.  Our proof relies crucially on a
recent result by \citet{demri12}, who show the
\textsc{2-ExpTime}-completeness of the coverability problem in
\emph{branching vector addition systems with states} (BVASS).  These
systems form a natural generalisation of vector addition systems, and
have been defined independently in a variety of contexts (see the
survey~\citep{acl/Schmitz10} and \autoref{sec-bvass} below), notably
that of provability in \emph{multiplicative exponential linear logic}
(\textbf{MELL}, see~\citep{vata}).  More precisely:
\begin{itemize}
\item In \autoref{sec-upb}, we show that so-called \emph{expansive}
  BVASSs can simulate proofs in $\Rel$ in a natural manner by
  exploiting the subformula property of its usual sequent calculus
  $L\Rel$.  We then show how to reduce reachability in expansive BVASS
  to coverability, thereby providing a decision procedure in doubly
  exponential time.
\item The matching hardness proof in \autoref{sec-lowb} relies on the
  one hand on \emph{comprehensive} instances of the BVASS coverability
  problem, and on the other hand on a new \emph{focusing} sequent
  calculus $F\Rel$ for $\Rel$.
\item The reduction from $\Rel$ provability to expansive BVASS
  reachability is actually a special case of a more general reduction
  proved in~\citep{lazic14} for the multiplicative exponential
  fragment of intuitionistic contractive linear logic
  (\textbf{IMELLC}), i.e.\ \textbf{IMELL} with structural contraction.
  Our reduction in \autoref{sub-exp-cov} from expansive reachability
  to coverability thus entails that \textbf{IMELLC} provability
  is \textsc{2-ExpTime}-complete, as explained in \autoref{sec-ext}.
\end{itemize}
Let us first recall the formal definition of $\Rel$
before turning to that of BVASSs in \autoref{sec-bvass}.

\section{The Implicational Fragment $\Rel$}\label{sec-relevance}
The reader will find in \citep[\sectionautorefname\ 4]{dunn02} a nice
overview of the decision problem for $\mathbf{R}$, covering in
particular \citeauthor{kripke59}'s solution for
$\Rel$~\citep{kripke59} and \citeauthor{urquhart99}'s lower bound
argument for $\mathbf{R}_{{\imp},{\wedge}}$~\citep{urquhart99}.

\subsection{A \detitle{Sequent Calculus}}
We recall here the formal definition of $\Rel$ as a sequent calculus
$L\Rel$ in Gentzen's style.  Let $\?{A}$ be a countable set of atomic
propositions; we define the set of formul\ae\ as following the
abstract syntax
\begin{equation}
  A ::= a\mid A\imp A\tag{implicational formul\ae}
\end{equation}
where $a$ ranges over $\?{A}$.  We consider $\imp$ to be
right-associative, e.g.\ $A\imp B\imp C$ denotes $A\imp (B\imp C)$.  In
the following rules, we use $A,B,C,\dots$ to denote implicational
formul\ae\ and $\Gamma,\Delta,\dots$ to denote multisets of such
formul\ae; commas in e.g.\ `$\Gamma, A$' and `$\Gamma,\Delta$' denote
multiset unions of $\Gamma$ with the singleton $A$ and with $\Delta$
respectively; finally, a \emph{sequent} is a pair `$\Gamma\vdash A$'
stating that the succedent $A$ is valid assuming the antecedent
$\Gamma$ to be relevant:
\begin{equation*}
  \setlength{\arraycolsep}{8pt}
  \begin{array}{cc} 
    \tabrulelabelr[identity]{rl-ax}{}{A\vdash A}{Id}&
  \tabrulelabelr[contraction]{rl-cn}{\Gamma, A, A\vdash B}{\Gamma,
    A\vdash B}{C}\\[1.5em]
  \tabrulelabelr[left implication]{rl-impl}{\Gamma\vdash A\quad\Delta, B\vdash
    C}{\Gamma,\Delta,A\imp B\vdash C}{\imp_L}
  &
  \tabrulelabelr[right implication]{rl-impr}{\Gamma, A\vdash B}{\Gamma\vdash A\imp B}{\imp_R}
\end{array}\end{equation*}
As we work with multisets, this sequent calculus includes implicitly
the structural `exchange' rule.  It does however not feature the
classical `weakening' rule---which would defeat the very point of
relevance---nor the `cut' rule---which is admissible.  A visible
consequence of this definition is that the calculus enjoys the
\emph{subformula property}: all the formul\ae\ in rule premises are
subformul\ae\ of the formul\ae\ appearing in the corresponding
consequences.

\subsection{Decidability and \detitle{Complexity}}
With hindsight, the decision procedure of \citet{kripke59} for
provability in the implicational fragment of relevance logic can be
seen as a precursor for many later algorithms that rely on the
existence of a \emph{well quasi ordering} (wqo) for their
termination~\citep{riche98}.  This decision procedure can be
understood as an application of Dickson's Lemma to prove the
finiteness of `irredundant' proof trees for the target sequent $\vdash
A$.  Furthermore, combinatorial analyses of Dickson's Lemma as e.g.\
in~\citep{FFSS-lics2011} provide explicit upper bounds on the size of
those irredundant proofs, in the form of the Ackermann function in the
size of $A$, yielding an \textsc{Ackermann} upper bound for $\Rel$
provability, as shown by \citet{urquhart99}.

Regarding lower bounds, \citeauthor{urquhart90}
in~\citep[\sectionautorefname~9]{urquhart90} explains how to derive
\textsc{Exp\-Space}-hardness for $\Rel$, using model-theoretic
techniques to reduce from the word problem for finitely presented
commutative semigroups~\citep{mayr82}.

\subsection{Strict {\boldmath $\lambda$}-\detitle{Calculus}}
The implicational fragment $\Rel$ is in bijection with the typing
rules of the simply typed $\lambda I$-calculus, where abstracted terms
$\lambda x.t$ are well-formed only if $x$ appears free in $t$; see
\citep[\sectionautorefname~9F]{curry58}.  This means that $\Rel$
provability can be restated as the \emph{type inhabitation} problem
for the simply typed $\lambda I$-calculus.  Our complexity results
should then be contrasted with the \textsc{PSpace}-completeness of the
same problem for the simply typed $\lambda$-calculus~\citep{statman79}.

\section{Branching VASS}\label{sec-bvass}
\emph{Branching vector addition systems with states} (hereafter BVASS) have been
independently defined in several contexts%
\ifshort
\else%
  :
\begin{itemize}
\item in computational linguistics~\citep{acl/Schmitz10}, as a means
  of modelling grammatical dependencies between tree nodes in
  \emph{unordered vector grammars with dominance
    edges}~\citep{uvgdl}---they have since been related with
  \emph{abstract categorial grammars}~\citep{acg} and \emph{minimalist
    grammars}~\citep{mgmell}---;
\item in computational logic, as a way to attack the decision problem for
  $\mathbf{MELL}$ through counter machines~\citep{vata},
\item in cryptographic protocol
  verification, as a means to reason in rewriting systems
  modulo associativity and commutativity of some
  symbols~\citep{bvass}.
\end{itemize}
They have furthermore been linked to the satisfiability problem of
two-variables FO over data trees~\citep{2vl,dimino13} and the modelling
of parallel programming libraries~\citep{bouajjani13}.
\fi

\subsection{Formal \detitle{Definitions}}\label{sub-bvass}
Given $d$ in $\+N$, we write `$\vec 0$' for the null vector in
$\+N^d$, and for $0<i\leq d$, `$\vec e_i$' for the unit vector in
$\+N^d$ with $1$ on coordinate $i$ and $0$ everywhere else.  Let
$U_d\eqdef\{\vec e_i,-\vec e_i\mid 0<i\leq d\}$.
Syntactically, an \emph{ordinary BVASS} is a tuple
$\?B=\tup{Q,d,T_u,T_s}$ where $Q$ is a finite set of \emph{states},
$d$ is a \emph{dimension} in $\+N$, and $T_u\subseteq Q\times
U_d\times Q$ and $T_s\subseteq Q^3$ are respectively finite sets
of \emph{unary} and \emph{split} rules.  We denote unary rules
$(q,\vec u,q_1)$ in $T_u$ with $\vec u$ in $U_d$ by
`$q\xrightarrow{\vec u}q_1$' and split rules $(q,q_1,q_2)$ in $T_s$ by
`$q\to q_1+q_2$'.

We define the semantics of an ordinary BVASS through a deduction system
over \emph{configurations} $(q,\vec v)$ in $Q\times\+N^d$:
\begin{equation*}
  \tabrulelabelr[increment]{bv-i}{q,\vec v}{q_1,\vec v+\vec e_i}{incr}
  \qquad
  \tabrulelabelr[decrement]{bv-d}{q,\vec v+\vec e_i}{q_1,\vec v}{decr}
  \qquad
  \tabrulelabelr[split]{bv-s}{q,\vec v_1+\vec v_2}{q_1,\vec v_1\quad q_2,\vec
    v_2}{split}
\end{equation*}
respectively for unary rules %
$q\xrightarrow{\vec e_i}q_1$ and $q\xrightarrow{-\vec e_i}q_1$ in
$T_u$ and a split rule $q\to q_1+q_2$ in $T_s$; in \eqref{bv-s} `$+$'
denotes component-wise addition in $\+N^d$.  Such a deduction system
can be employed either \emph{top-down} or \emph{bottom-up} depending
on the decision problem at hand (as with tree automata); the top-down
direction will correspond in a natural way to goal-directed
\emph{proof search} in the sequent calculus of
\autoref{sec-relevance}.

\subsubsection{Ordinary BVASSs\nopunct} are a slight restriction over
BVASSs, which would in general allow any vector in $\+Z^d$ in unary
rules.  Because they often lead to more readable proofs, we only
employ ordinary BVASSs in this paper.  This is at no loss of
generality, since one can build an ordinary BVASS `equivalent' to a
given BVASS in logarithmic space, where equivalence should be
understood relative to the reachability and coverability problems;
see \citep{acl/Schmitz10} for details.

\subsubsection{Reachability.\nopunct} 
Branching VASSs are associated with a natural decision problem:
\emph{reachability} asks, given a BVASS $\?B$, a root state $q_r$, and
a leaf state $q_\ell$, whether there exists a deduction tree with root
label $(q_r,\vec 0)$ and every leaf labelled $(q_\ell,\vec 0)$; such a
deduction tree is called a \emph{reachability witness}.  De Groote et
al.~\citep{vata} have shown that this problem is recursively
equivalent to $\mathbf{MELL}$ provability, and it is currently unknown
whether it is decidable---both problems are however known to be of
non-elementary computational complexity~\citep{lazic14}.

Let us introduce some additional notation that will be handy in
proofs.  We write `$\?B,T,q_\ell\jdg q,\vec v$' if there exists a
deduction tree of $\?B$ with root label $(q,\vec v)$ and leaves
labelled by $(q_\ell,\vec 0)$, which uses each rule in $T\subseteq
T_u\uplus T_s$ at least once.  Such \emph{root judgements} can be
derived through the deduction system
\begin{gather*}
   \frac{}{\?B,\emptyset,q_\ell\jdg q_\ell,\vec 0}
   \qquad
   \frac{\?B,T,q_\ell\jdg q_1,\vec v+\vec e_i}
        {\?B,T\cup\{q\xrightarrow{\vec e_i}q_1\},q_\ell\jdg q,\vec v}
   \\[1em]
   \frac{\?B,T,q_\ell\jdg q_1,\vec v}
        {\?B,T\cup\{q\xrightarrow{-\vec e_i}q_1\},q_\ell\jdg q,\vec v+\vec e_i}
   \qquad
   \frac{\?B,T_1,q_\ell\jdg q_1,\vec v_1
         \quad
         \?B,T_2,q_\ell\jdg q_2,\vec v_2}
    {\?B,T_1\cup T_2\cup\{q\to q_1+q_2\},q_\ell\jdg q,\vec v_1+\vec v_2}
\end{gather*}
We write more simply `$\?B,q_\ell\jdg q,\vec v$' if there exists
$T\subseteq T_u\uplus T_s$ such that $\?B,T,q_\ell\jdg q,\vec v$.
With these notations, the reachability problem asks whether
$\?B,q_\ell\jdg q_r,\vec 0$.

\subsection{Root Coverability}\label{sub-cov}
Our interest in this paper lies in a relaxation of the reachability
problem, where we ask instead to \emph{cover} the root: given as
before $\tup{\?B,q_r,q_\ell}$, we ask whether there exists
a \emph{coverability witness}, i.e.\ a deduction tree with root
$(q_r,\vec v)$ for some $\vec v$ in $\+N^d$ and leaves
$(q_\ell,\vec 0)$; in other words whether $\?B,q_\ell\jdg q_r,\vec
v$ for some $\vec v$ in $\+N^d$.

This problem was shown decidable by \citet{bvass}, and was later
proven \textsc{2-ExpTime}-complete by \citet{demri12} in a slight
variant called \emph{branching vector addition systems} (BVAS):
\begin{fact}[\citenum{demri12}, \theoremautorefname~8 and \theoremautorefname~21]\label{fc-bvass}
  BVAS coverability is \textsc{2-ExpTime}-complete.
\end{fact}

Branching VAS are not equipped with a state space $Q$.  Their
coverability problem is stated slightly differently, but is easy to
reduce in both directions to BVASS coverability.  This would not be
worth mentioning here if it were not for the following instrumental
corollary of their proof, which exploits an encoding of
$d$-dimensional BVASSs into $(d+6)$-dimensional BVASs:
\begin{restatable}{corollary}{corbvass}\label{cor-bvass}
  Coverability in a BVASS $\?B=\tup{Q,d,T_u,T_s}$ can be solved in
  deterministic time $2^{2^{O(n\cdot\log (n\cdot\log |Q|))}}$, where $n$ denotes
  the size of the representation of $\tup{d,T_u}$.
\end{restatable}
\noindent
Corollary~\ref{cor-bvass} entails that coverability remains in
\textsc{2-ExpTime} for BVASSs with double exponential state space.
This is an easy result, which we show in \appref{sec-bvas}.

\section{Upper Bound}\label{sec-upb}
In order to show a \textsc{2-ExpTime} upper bound for $\Rel$
provability, we introduce as an intermediate decision problem
the \emph{expansive reachability} problem for BVASS
(\autoref{sub-exp}).  Then, the first step of our proof
in \autoref{sub-rel-bvass} takes us from the sequent calculus $L\Rel$
to reachability in expansive BVASS.  This is a simple construction
that relies on the subformula property of $L\Rel$, and is actually a
particular case of a more general reduction shown
in \citep[\propositionautorefname~9]{lazic14}.  The new technical
result here is the second step: a reduction from expansive BVASS
reachability to BVASS coverability, which is shown
in \autoref{sub-exp-cov}.  This new reduction also entails new upper
bounds for provability in extensions of $L\Rel$ studied
in \citep{lazic14}; see \autoref{sec-ext}.

\subsection{Expansive Reachability}\label{sub-exp}
An \emph{expansive BVASS} is a BVASS with an additional deduction
rule:
\begin{equation*}
   \tabrulelabelr[expansion]{bv-e}{q,\vec v+\vec e_i}{q,\vec v+2\vec e_i}{expansion}
\end{equation*}
Note that expansions could be simulated by unary rules
$q\xrightarrow{-\vec e_i}q_i\xrightarrow{\vec
e_i}q'_i\xrightarrow{\vec e_i}q$ for all $q$ in $Q$ and $0<i\leq d$;
we prefer to see them as new deduction rules.

This yields a new rule for root judgements, which we denote
using `$\jdg_e$' to emphasise that we allow expansion rules:
\begin{equation*}
   \frac{\?B,T,q_\ell\jdg_e q,\vec v+2\vec e_i}{\?B,T,q_\ell\jdg_e
   q,\vec v+\vec e_i}
\end{equation*}
The \emph{expansive reachability} problem then asks, given an
expansive BVASS $\?B$ and two states $q_r$ and $q_\ell$, whether
$\?B,q_\ell\jdg_e q_r,\vec 0$.

\subsection{From {\boldmath $L\Rel$} to Expansive Reachability}
\label{sub-rel-bvass}
We prove here the following reduction:
\begin{proposition}\label{prop-rel-bvass}
  There is a logarithmic space reduction from provability in $\Rel$ to
  expansive reachability in ordinary BVASSs.
\end{proposition}
Let us consider an instance $\tup{F}$ of the provability problem for
$\Rel$ for an implicational formula $F$.  The instance is positive if
and only if we can find a proof for $\vdash F$ in $L\Rel$.  Thanks to
the subformula property, we know that in such a proof, all the
sequents $\Gamma\vdash A$ must use subfomul\ae\ of $F$.  That is, if we
denote by $S$ the set of subformul\ae\ of $F$, then $\Gamma$ is in
$\+N^S$ and $A$ in $S$.

We construct from $F$ an expansive BVASS $\?B_F$ that implements proof
search in $L\Rel$ restricted to subformul\ae\ of $F$.  We define for
this $\?B_F\eqdef\tup{Q_F,|S|,T_u,T_s}$ where the state space $Q_F$
includes $S$ and a distinguished leaf state $q_\ell$.  It also includes
some intermediate states as introduced in the translations of the
rules of $L\Rel$ into rules in $T_u\cup T_s$ depicted
in \autoref{fig-rel-bvass}.
\begin{figure}[tbp]
  \centering \begin{tikzpicture}[auto,node distance=.65cm]
    \node[state,label=left:{\ref{rl-ax}:}](init0){$A$};
    \node[state,accepting by double,right=of init0](init1){$q_\ell$};   
    \path (init0) edge node{$-\vec e_{A}$} (init1);
    \node[state,right=1.3cm of init1,label=left:{\ref{rl-impl}:}](Limp0){$C$};
    \node[intermediate,right=1.1cm of Limp0,label=right:{$+$}](Limp1){};
    \draw[draw=black!40] (Limp1) ++(.2,-.2) arc (-40:40:.33);
    \node[state,above right=of Limp1](Limp2){$A$};
    \node[intermediate,below right=of Limp1](Limp3){};
    \node[state,right=of Limp3](Limp4){$C$};
    \path (Limp0) edge node{$-\vec e_{A\imp B}$} (Limp1)
      (Limp1) edge (Limp2)
      (Limp1) edge (Limp3)
      (Limp3) edge node{$\vec e_B$}(Limp4);
    \node[rstate,right=3.3cm of
      Limp1,label=left:{\ref{rl-impr}:}](Rimp0){$A\imp B$};
    \node[state,right=of Rimp0](Rimp1){$B$};
    \path (Rimp0) edge node{$\vec e_A$} (Rimp1);
  \end{tikzpicture}
  \caption{\label{fig-rel-bvass}The rules and intermediate states of $\?B_F$.}
\end{figure}
Note that \eqref{rl-cn} has no associated rule; it relies instead on
expansions in $\?B_F$.  The full state space $Q_F$ of $\?B_F$,
including intermediate states, is thus of size $O(|F|^2)$.

Let us write $\vec v_\Gamma$ for the vector in $\+N^{|S|}$ associated
with a multiset $\Gamma$ in $\+N^S$.  The proof of
\autoref{prop-rel-bvass} is a consequence of the following claim
instantiated with $A=F$ and $\Gamma=\emptyset$:
\begin{restatable}{claim}{clrelbvass}\label{cl-rel-bvass}%
  For all $\Gamma$ in $\+N^S$ and $A$ in $S$, $\Gamma\vdash A$
  if and only if $\?B_F,q_\ell\jdg_e A,\vec v_\Gamma$.%
\end{restatable}
\ifshort
\else%
\begin{proof}
  We proceed by induction of the structure of proofs in $L\Rel$ and of
  expansive root judgements in $\?B_F$.
  \begin{description}
  \item[\nameref{rl-ax}] for the base case~\eqref{rl-ax}, $A\vdash A$
    iff $\?B_F,q_\ell\jdg_e A,\vec e_A$.
  \item[\nameref{rl-cn}] If $\Gamma,A\vdash B$ as the result
    of~\eqref{rl-cn} from $\Gamma,A,A\vdash B$ then by induction
    hypothesis $\?B_F,q_\ell\jdg_e B,\vec v_{\Gamma,A,A}$, from
    which an expansion yields $\?B_F,q_\ell\jdg_e B,\vec
    v_{\Gamma,A}$.  Conversely, if $\?B_F,q_\ell\jdg_e B,\vec
    v_{\Gamma,A}$ as the result of an expansion from
    $\?B_F,q_\ell\jdg_e B,\vec v_{\Gamma,A,A}$, then by induction
    hypothesis $\Gamma,A,A\vdash B$ and a contraction yields
    $\Gamma,A\vdash B$.
  \item[\nameref{rl-impl}] If $\Gamma,\Delta,A\imp B\vdash C$ as the
    result of~\eqref{rl-impl} applied to $\Gamma\vdash A$ and
    $\Delta,B\vdash C$, then by induction hypothesis
    $\?B_F,q_\ell\jdg_e A,\vec v_\Gamma$ and $\?B_F,q_\ell\jdg_e
    C,\vec v_{\Delta,B}$.  Then the group of rules in $\?B_F$ yields
    $\?B_F,q_\ell\jdg_e C_B,\vec v_{\Delta}$ and
    $\?B_F,q_\ell\jdg_e C_{A\imp B},\vec v_{\Gamma,\Delta}$ in the
    intermediate states, and finally $\?B_F,q_\ell\jdg_e C,\vec
    v_{\Gamma,\Delta,A\imp B}$ as desired.  The converse direction is
    similar.
  \item[\nameref{rl-impr}] If $\Gamma\vdash A\imp B$ as the result
    of~\eqref{rl-impr} applies to $\Gamma,A\vdash B$, then by
    induction hypothesis $\?B_F,q_\ell\jdg_e B,\vec v_{\Gamma,A}$,
    from which the corresponding rule of $\?B_F$ yields
    $\?B_F,q_\ell\jdg_e A\imp B,\vec v_{\Gamma}$ as desired.  The
    converse direction is similar.\qedhere
  \end{description}
\end{proof}

\fi

\subsection{From Expansive Reachability to Coverability}
\label{sub-exp-cov}
The second step of our proof that $\Rel$ provability is in
\textsc{2-ExpTime} is then to reduce expansive reachability to
coverability in BVASS.  Our reduction incurs an exponential blow-up in
the number of states, but thanks to \autoref{cor-bvass}, this still
results in a \textsc{2-ExpTime} algorithm:

\begin{proposition}\label{prop-exp-cov}
  There is a polynomial space reduction from BVASS expansive
  reachability to BVASS coverability.
\end{proposition}

\subsubsection{Topmost \detitle{Increments}.\nopunct}
Consider an instance $\tup{\?B,q_r,q_\ell}$ of the expansive
reachability problem with $\?B=\tup{Q,d,T_u,T_s}$.  Because the root
vector of an expansive reachability witness must be $\vec 0$, we can
identify along each branch of the witness and for each coordinate
$0<i\leq d$ the topmost (i.e.\ closest to the root) application of
an~\eqref{bv-i} rule---possibly no such increment ever occurs on some
branches.

Assume without loss of generality that $q_\ell$ has no outgoing
transition in $\?B$.  We construct a new BVASS
$\?B^\dagger=\tup{Q^\dagger,d,T^\dagger_u,T_s}$ with additional states
$q_\ell^i$ and unary rules $q_\ell\xrightarrow{\vec
  e_i}q_\ell^i\xrightarrow{-\vec e_i}q_\ell$ for every $0<i\leq d$.
Then $\?B,q_\ell\jdg_e q_r,\vec 0$ if and only if
$\?B^\dagger,q_\ell\jdg_e q_r,\vec 0$ (observe in particular that no
expansion in $q_\ell^i$ can occur in an expansive reachability
witness).  Additionally, the new rules allow us to assume that there
\emph{is} a topmost increment for each branch and every coordinate of
an expansive reachability witness of $\?B^\dagger$.

Let $[d]\eqdef\{1,\dots,d\}$.  The root judgement relation can be
refined as `$\jdg_e^s$' with a set $s\subseteq[d]$ of coordinates.
The intended semantics for $i\in s$ is that there is at least one
increment on coordinate $i$ earlier on the path from the root in the
expansive reachability witness.  Formally, at the leaves
\begin{gather*}
  \frac{}{\?B^\dagger,\emptyset,q_\ell\jdg_e^{[d]}q_\ell,\vec 0}
\shortintertext{since by assumption every coordinate must see an
increase.  Then, an increment is either topmost or not:}
  \frac{\?B^\dagger,T,q_\ell\jdg_e^{s\uplus\{i\}}q_1,\vec w+\vec e_i}
       {\?B^\dagger,T\cup\{q\xrightarrow{\vec e_i}q_1\},q_\ell\jdg_e^{s}q,\vec w}
  \qquad 
  \frac{\?B^\dagger,T,q_\ell\jdg_e^{s\cup\{i\}}q_1,\vec v+\vec e_i}
       {\?B^\dagger,T\cup\{q\xrightarrow{\vec e_i}q_1\},q_\ell\jdg_e^{s\cup\{i\}}q,\vec v}
\shortintertext{where $\vec w(i)=0$ and `$\uplus$' denotes disjoint
  union.  Decrements and expansions are necessarily dominated by the
  topmost increment:}
  \frac{\?B^\dagger,T,q_\ell\jdg_e^{s\cup\{i\}}q_1,\vec v}
       {\?B^\dagger,T\cup\{q\xrightarrow{-\vec e_i}q_1\},q_\ell\jdg_e^{s\cup\{i\}}q,\vec v+\vec e_i}
  \qquad
  \frac{\?B^\dagger,T,q_\ell\jdg_e^{s\cup\{i\}}q,\vec v+2\vec e_i}
       {\?B^\dagger,T,q_\ell\jdg_e^{s\cup\{i\}}q,\vec v+\vec e_i}
\shortintertext{Finally, the same topmost increments have been seen on
  both branches of a split:}
  \frac{\?B^\dagger,T_1,q_\ell\jdg_e^s q_1,\vec v_1
   \quad\?B^\dagger,T_2,q_\ell\jdg_e^s q_2,\vec v_2}
       {\?B^\dagger,T_1\cup T_2\cup\{q\to q_1+q_2\},q_\ell\jdg_e^s
  q,\vec v_1+\vec v_2}
\end{gather*}

The refined root judgements verify
\begin{equation}\label{eq-topmost}
  \?B,q_\ell\jdg_e q_r,\vec 0\text{ implies
  }\?B^\dagger,q_\ell\jdg_e^\emptyset q_r,\vec 0\;,
\end{equation}
the converse implication being immediate by removing the `$s$'
annotations.

\subsubsection{Reduction to \detitle{Coverability}.\nopunct}
We construct yet another BVASS $\?B^\ddagger=\tup{Q^\dagger\times
2^{[d]},d,T^\ddagger_u,T^\ddagger_s}$ and build a coverability
instance $\tup{\?B^\ddagger,(q_r,\emptyset),(q_\ell,[d])}$.  The idea
is to maintain a set $s\subseteq[d]$ as in the refined judgements
$\jdg_e^s$; however since we cannot test for zero we will rely
instead on nondeterminism.  Let
\begin{align}
  T_u^\ddagger
  &\eqdef%
  \label{eq-covi}\tag{incr$^\ddagger$}
        \{(q,s)\xrightarrow{\vec e_i}(q_1,s\cup\{i\})
        \mid q\xrightarrow{\vec e_i}q_1\in T_u^\dagger,s\subseteq[d]\}\\
  \label{eq-covd}\tag{decr$^\ddagger$}
  &\:\cup\:\{(q,s\cup\{i\})\xrightarrow{-\vec e_i}(q_1,s\cup\{i\})
        \mid q\xrightarrow{-\vec e_i}q_1\in T_u^\dagger,s\subseteq[d]\}\;,\\
  \label{eq-covs}\tag{split$^\ddagger$}
  T_s^\ddagger&\eqdef\{(q,s)\to(q_1,s)+(q_2,s)\mid q\to q_1+q_2\in
T_s^\dagger,s\subseteq[d]\}\;.
\end{align}

\begin{restatable}{claim}{clexpcova}\label{cl-exp-cov1}
  If $\?B^\dagger,q_\ell\jdg_e^s,q,\vec v$, then there exists $\vec
  v'\geq\vec v$ such that $\?B^\ddagger,(q_\ell,[d])\jdg (q,s),\vec
  v'$.
\end{restatable}%
\ifshort\else%
\begin{proof}%
  We prove the claim by induction over the root judgement in
$\?B^\dagger$.  For the base case with $q=q_\ell$ and $\vec v=\vec 0$,
we choose $\vec v'=\vec 0$.

For the induction step, if the judgement results from the application
of an increment rule on coordinate $i$, then
$\?B^\dagger,q_\ell\jdg_e^{s\cup\{i\}}q_1,\vec v+\vec e_i$, and by
induction hypothesis there exists $\vec v'\geq v$ such that
$\?B^\ddagger,(q_\ell,[d])\jdg (q_1,s\cup\{i\}),\vec v'+\vec e_i$.
Two cases arise depending on whether $i\not\in s$ or $i\in s$, i.e.\
whether this is the topmost increment on coordinate $i$ or not.  In
both cases, \eqref{eq-covi} yields $\?B,(q_\ell,[d])\jdg (q,s),\vec
v'$ as desired.

If the judgement results from the application of a decrement to
$\?B^\dagger,q_\ell\jdg_e^{s\cup\{i\}}q_1,\vec v$, then by induction
hypothesis there exists $\vec v'\geq\vec v$ such that
$\?B^\ddagger,(q_\ell,[d])\jdg (q_1,s\cup\{i\}),\vec v'$,
and~\eqref{eq-covd} shows $\?B^\ddagger,(q_\ell,[d])\jdg
(q,s\cup\{i\}),\vec v'+\vec e_i$ where $\vec v'+\vec e_i\geq \vec
v+\vec e_i$ as desired.

If the judgement results from an expansion applied to
$\?B^\dagger,q_\ell\jdg_e^{s\cup\{i\}}q,\vec v+2\vec e_i$, then by
induction hypothesis there exists $\vec v'\geq\vec v+2\vec e_i\geq\vec
v+\vec e_i$ and $\?B^\ddagger,(q_\ell,[d])\jdg (q,s\cup\{i\}),\vec
v'$ as desired.

Finally, if the judgement results from a split with premises
$\?B^\dagger,q_\ell\jdg_e^s q_1,\vec v_1$ and
$\?B^\dagger,q_\ell\jdg_e^s q_2,\vec v_2$, then by induction
hypothesis there exist $\vec v'_1\geq \vec v_1$ and $\vec v'_2\geq\vec
v_2$ such that $\?B^\ddagger,(q_\ell,[d])\jdg(q_1,s),\vec v'_1$ and
$\?B^\ddagger,(q_\ell,[d])\jdg(q_2,s),\vec v'_2$.  Thanks
to~\eqref{eq-covs}, this yields
$\?B^\ddagger,(q_\ell,[d])\jdg(q,s),\vec v'_1+\vec v'_2$ where
$\vec v'_1+\vec v'_2\geq\vec v_1+\vec v_2$ as desired.
\end{proof}%
\fi%
\ifshort\else
For the converse direction, given $s\subseteq[d]$ and $\vec v$ in
$\+N^d$, we define $s\cdot\vec v$ for each $0<i\leq d$ by
\begin{equation}
  (s\cdot\vec v)(i)\eqdef\begin{cases}\vec v(i)&\text{if }i\in s\;,\\
  0&\text{otherwise}.\end{cases}
\end{equation}
\fi
\begin{restatable}{claim}{clexpcovb}\label{cl-exp-cov2}
  If $\?B^\ddagger,(q_\ell,[d])\jdg (q,s),\vec v$, then
  $\?B^\dagger,q_\ell\jdg_e q,s\cdot\vec v$.
\end{restatable}%
\ifshort\else%
\begin{proof}%
  We prove the claim by induction on the root judgement in
$\?B^\ddagger$.  For the base case, $\?B^\ddagger,(q_\ell,[d])\jdg
(q_\ell,[d]),\vec 0$ indeed matches $\?B^\dagger,q_\ell\jdg_e
q_\ell,[d]\cdot\vec 0$.

For the induction step, first assume that the judgement stems
from~\eqref{eq-covi} and $\?B^\ddagger,(q_\ell,[d])\jdg
(q_1,s\cup\{i\}),\vec v+\vec e_i$.  By induction hypothesis,
$\?B^\dagger,q_\ell\jdg_e q_1,(s\cup\{i\})\cdot(\vec v+\vec e_i)$.
Two cases arise depending on whether $i\in s$.  First suppose $i\in
s$: then $(s\cup\{i\})\cdot(\vec v+\vec e_i)=s\cdot\vec v+\vec e_i$
and $q\xrightarrow{\vec e_i}q_1$ yields $\?B^\dagger,q_\ell\jdg_e
q,s\cdot\vec v$ as desired.  Assuming $i\not\in s$, let us decompose
$\vec v$ as $\vec w+n\vec e_i$ where $\vec w(i)=0$ and $n\geq 0$: then
$(s\cup\{i\})\cdot(\vec v+\vec e_i)=s\cdot\vec w+(n+1)\vec e_i$.  We
apply $n$ expansions on coordinate $i$ to show
$\?B^\dagger,q_\ell\jdg_e q_1,s\cdot\vec w+\vec e_i$.  Applying
$q\xrightarrow{\vec e_i}q_1$ then yields $\?B^\dagger,q_\ell\jdg_e
q,s\cdot\vec w$, where $s\cdot\vec w=s\cdot\vec v$ as desired.

Assume now that~\eqref{eq-covd} was applied to
$\?B^\ddagger,(q_\ell,[d])\jdg(q_1,s\cup\{i\}),\vec v$.  By
induction hypothesis, $\?B^\dagger,q_\ell\jdg_e
q_1,(s\cup\{i\})\cdot\vec v$, and $q\xrightarrow{-\vec e_i}q_1$ yields
$\?B^\dagger,q_\ell\jdg_e q,(s\cup\{i\})\cdot\vec v+\vec e_i$, where
$(s\cup\{i\})\cdot\vec v+\vec e_i=(s\cup\{i\})\cdot(\vec v+\vec e_i)$
as desired.

Finally, assume that~\eqref{eq-covs} was applied to
$\?B^\ddagger,(q_\ell,[d])\jdg (q_1,s),\vec v_1$ and to
$\?B^\ddagger,(q_\ell,[d])\jdg(q_2,s),\vec v_2$.  By induction
hypothesis, $\?B^\dagger,q_\ell\jdg_e q_1,s\cdot\vec v_1$ and
$\?B^\dagger,q_\ell\jdg_e q_2,s\cdot\vec v_2$, from which $q\to
q_1+q_2$ yields $\?B^\dagger,q_\ell\jdg_e q,(s\cdot\vec
v_1)+(s\cdot\vec v_2)$, where $(s\cdot\vec v_1)+(s\cdot\vec
v_2)=s\cdot(\vec v_1+\vec v_2)$ as desired.
\end{proof}%
\fi%

\begin{proof}[Proof of \autoref{prop-exp-cov}]
If $\?B,q_\ell\jdg_e q_r,\vec 0$, then by~\eqref{eq-topmost},
$\?B^\dagger,q_\ell\jdg_e^\emptyset q_r,\vec 0$, thus
by \autoref{cl-exp-cov1}, there exists $\vec v$ such that
$\?B^\ddagger,(q_\ell,[d])\jdg (q_r,\emptyset),\vec v$, i.e.\ we can
cover $(q_r,\emptyset)$ in $\?B^\ddagger$.

Conversely, if $\?B^\ddagger,(q_\ell,[d])\jdg (q_r,\emptyset),\vec
v$, then by \autoref{cl-exp-cov2}, $\?B^\dagger,q_\ell\jdg_e
q_r,\emptyset\cdot\vec v$ where $\emptyset\cdot\vec v=\vec 0$.
Therefore $\?B,q_\ell\jdg_e q_r,\vec 0$ in the original BVASS $\?B$.
\end{proof}

\begin{theorem}\label{th-upb}
  Provability in $\Rel$ is in \textsc{2-ExpTime}.
\end{theorem}
\begin{proof}
  By~\autoref{prop-rel-bvass} and \autoref{prop-exp-cov}, from a
  provability instance $\tup{F}$, we can reduce to a coverability
  instance $\tup{\?B_F^\ddagger,(q_r,\emptyset),(q_\ell,[|F|])}$ where
  $\?B_F^\ddagger$ has dimension $|F|$ and a number of states in
  $2^{p(|F|)}$ for a polynomial $p$.  By \autoref{cor-bvass},
  this coverability instance can be solved in double exponential time
  in $|F|$.  Note that the coverability check can be performed
  on-the-fly from $F$ to avoid the explicit construction of
  $\?B_F^\ddagger$.
\end{proof}

\section{Lower Bound}\label{sec-lowb}
In this section, we exhibit a reduction from BVASS coverability to
$\Rel$ provability, thereby showing its \textsc{2-ExpTime}-hardness.

Previous reductions from counter machines to substructural logics
in~\citep{lincoln92,urquhart99,lazic14} actually reduce to provability
in the logic extended with a \emph{theory} encoding the rules of the
system, which is then reduced to the basic logic.  This last step
relies in an essential way on the presence of exponential or additive
connectives to `dispose' of unused rules.

Having neither exponential nor additive connectives at our disposal,
we introduce in \autoref{sub-compr} a \emph{comprehensive} variant of
the expansive reachability problem, where every rule should be
employed at least once in the deduction.  We further avoid the use of
a theory and define in \autoref{sub-focus} a \emph{focusing} calculus
for $\Rel$, from which the correctness of the reduction given in
\autoref{sub-hard} will be facilitated.

\subsection{Comprehensive \detitle{Reachability}}\label{sub-compr}
Given a BVASS $\?B=\tup{Q,d,T_u,T_s}$ with expansive semantics and two
states $q_r$ and $q_\ell$, the \emph{comprehensive reachability}
problem asks whether there exists a deduction tree of $\?B$ with root
label $(q_r,\vec 0)$ and leaves label $(q_\ell,\vec 0)$, such that
every rule in $T_u\cup T_s$ is used at least once.  In terms of root
judgements, it asks whether $\?B,T_u\cup T_s,q_\ell\jdg_e q_r,\vec 0$.
We show that BVASS coverability can be reduced to comprehensive
expansive reachability, hence by \autoref{fc-bvass}:
\begin{proposition}\label{prop-compr}
  Comprehensive reachability in expansive ordinary BVASS is
  \textsc{2-ExpTime}-hard.
\end{proposition}

\subsubsection{Increasing \detitle{Reachability}.\nopunct}
Let us consider an instance $\tup{\?B,q_r,q_\ell}$ of the coverability
problem in an ordinary BVASS $\?B=\tup{Q,d,T_u,T_s}$.  As a first
step, we construct an ordinary BVASS
$\?B^\dagger\eqdef\tup{Q,d,T^\dagger_u,T_s}$ with additional
\emph{increases} $q\xrightarrow{\vec e_i}q$ for every $q$ in $Q$ and
$0<i\leq d$.  We claim that coverability in $\?B$ is equivalent to
reachability in $\?B^\dagger$:
\begin{claim}\label{cl-cov-inc}
  There exists $\vec v$ in $\+N^d$ such that $\?B,q_\ell\jdg
  q_r,\vec v$ \ifshort\else if and only if\fi\ $\?B^\dagger,q_\ell\jdg q_r,\vec 0$.
\end{claim}
\begin{proof}[Proof Sketch]
Clearly, if $\?B,q_\ell\jdg q_r,\vec v$ for some $\vec v$ in
$\+N^d$, then $\?B^\dagger,q_\ell\jdg q_r,\vec v$, and using
increases in $q_r$ shows $\?B^\dagger,q_\ell\jdg q_r,\vec 0$.
Conversely, if there is a reachability witness for
$\tup{\?B^\dagger,q_r,q_\ell}$, then we can assume that increases
$q\xrightarrow{\vec e_i}q$ occur as close to the root as possible.  As
increases occurring right below increments, decrements, or splits can
be permuted locally to occur right above, such a reachability witness
has all its increases at the root.  The deduction tree below those
increases is labelled $(q_r,\vec v)$ for some $\vec v$ in $\+N^d$ and is
also a deduction tree of $\?B$.
\end{proof}

\subsubsection{Comprehensive \detitle{Root Rules}.\nopunct}
The second step of the reduction from BVASS coverability builds an
ordinary BVASS $\?B^\ddagger\eqdef\tup{Q^\ddagger,d+1+|T^\dagger_u\cup
  T_s|,T^\ddagger_u,T_s}$ where $Q^\ddagger\eqdef
Q\uplus\{q^\ddagger,q_r^\ddagger\}\uplus\{q_t\mid t\in T^\dagger_u\cup
T_s\}$.  It features an additional set of unary `root' rules---depicted
in \autoref{fig-compr}---designed to allow any rule in
$T^\ddagger_u\cup T_s$ to be employed in a reachability witness.

\begin{figure}[tbp]
\centering
\begin{tikzpicture}[auto,node distance=.65cm]
  \node[state](qrdd){$q^\ddagger_r$};
  \node[state,right=of qrdd](qdd){$q^\ddagger$};
  \node[state,right=.8cm of qdd](qr){$q_r$};
  \path (qrdd) edge node{$\vec e_{d+1}$} (qdd)
    (qdd) edge node{$-\vec e_{d+1}$} (qr);
  \node[state,right=1.3cm of qr](qi0){$q^\ddagger$};
  \node[state,right=of qi0,label=above:{$t$:}](qi1){$q$};
  \node[state,right=of qi1](qi2){$q_1$};
  \node[state,right=of qi2](qi3){$q_t$};
  \node[state,right=of qi3](qi4){$q^\ddagger$};
  \path (qi0) edge node{$\vec e_t$} (qi1)
    (qi1) edge node{$\vec e_i$} (qi2)
    (qi2) edge node{$-\vec e_i$} (qi3)
    (qi3) edge node{$-\vec e_t$} (qi4);
  \node[dotted,draw=black!40,rectangle,rounded
  corners=8pt,anchor=above left,minimum
  width=2.2cm,minimum height=1.2cm,above right=-.7cm and -.7cm of qi1]{};
  \node[state,below=.6cm of qrdd](qd0){$q^\ddagger$};
  \node[state,right=of qd0](qd1){$q_t$};
  \node[state,right=of qd1,label=above:{$t$:}](qd2){$q$};
  \node[state,right=of qd2](qd3){$q_1$};
  \node[state,right=of qd3](qd4){$q^\ddagger$};
  \path (qd0) edge node{$\vec e_t$} (qd1)
    (qd1) edge node{$\vec e_i$} (qd2)
    (qd2) edge node{$-\vec e_i$} (qd3)
    (qd3) edge node{$-\vec e_t$} (qd4);
  \node[dotted,draw=black!40,rectangle,rounded
  corners=8pt,anchor=above left,minimum
  width=2.2cm,minimum height=1.2cm,above right=-.7cm and -.7cm of qd2]{};
  \node[state,below=.26cm of qd4](qs0){$q^\ddagger$};
  \node[state,right=of qs0](qs1){$q_t$};
  \node[state,right=of qs1,label=right:{$+$},label=above:{$t$:}](qs2){$q$};
  \draw[draw=black!40] (qs2) ++(.3,-.3) arc (-45:45:.42);
  \node[state,above right=of qs2](qs3){$q_1$};
  \node[state,below right=of qs2](qs4){$q_2$};
  \node[state,below right=of qs3](qs5){$q^\ddagger$};
  \path (qs0) edge node{$\vec e_t$} (qs1)
    (qs1) edge node {$\vec e_t$} (qs2)
    (qs2) edge (qs3)
    (qs2) edge (qs4)
    (qs3) edge node {$-\vec e_t$} (qs5)
    (qs4) edge[swap] node {$-\vec e_t$} (qs5);
  \node[dotted,draw=black!40,rectangle,rounded
  corners=8pt,anchor=above left,minimum
  width=1.8cm,minimum height=2.7cm,right=-.8cm of qs2]{};
\end{tikzpicture}
\caption{\label{fig-compr}The rules of $\?B^\ddagger$ in the proof of
  \autoref{prop-compr}, where $t$ ranges over $T^\dagger_u\cup T_s$.}
\end{figure}

The idea is to introduce a new state $q^\ddagger$ and a new coordinate
for each rule $t$ in $T^\dagger_u\cup T_s$.  Starting from
$q^\ddagger$, $\?B^\ddagger$ can simulate any rule $t$ from
$T^\dagger_u\cup T_s$ by first incrementing by the corresponding unit
vector $\vec e_t$, then applying the rule, and finally decrementing by
$\vec e_t$ to return to $q^\ddagger$.  This ensures that, if
$\?B^\ddagger,T,q_\ell\jdg_e q^\ddagger,\vec v$ for some $\vec v$ in
$\+N^d$, then $\?B^\ddagger,T',q_\ell\jdg_e q^\ddagger,\vec v$ where
$T^\dagger_u\cup T_s\subseteq T'$.  The additional states and rules
from $q^\ddagger_r$ and to $q_r$ then show that:
\begin{claim}\label{cl-compr1}
  If $\?B^\dagger,T,q_\ell\jdg q_r,\vec 0$ for some $T\subseteq
  T^\dagger_u\cup T_s$, then $\?B^\ddagger,T^\ddagger_u\cup
  T_s,q_\ell\jdg_e q^\ddagger_r,\vec 0$.
\end{claim}

Conversely, assume that $\?B^\ddagger,q_\ell\jdg_e q^\ddagger_r,\vec
0$.  This entails $\?B^\ddagger,q_\ell\jdg_e q^\ddagger,\vec
e_{d+1}$.

First assume that no decrement by $\vec e_t$ for any $t$ in
$T^\dagger_u\cup T_s$ is ever used in the corresponding expansive
reachability witness.\ Then also no increment by
$\vec e_t$ occurs, and the only rule applicable at this point is
$q^\ddagger\xrightarrow{-\vec e_{d+1}}q_r$.  This yields a node $n$
labelled $(q_r,\vec 0)$, and a deduction subtree rooted by $n$ where
only rules of $\?B^\dagger$ and expansions are applied.  Because any
expansion can be simulated in $\?B^\dagger$ using an increase, this
yields $\?B^\dagger,q_\ell\jdg q_r,\vec 0$.

Assume now the opposite: there is at least one occurrence of a
decrement by $\vec e_t$ in the expansive reachability witness.
Consider the bottommost such occurrence along some branch, necessarily
yielding a node with $q^\ddagger$ as state label.  Then in the same
way, below this bottommost occurrence, no increment by $\vec e_t$
occurs, and the only rule applicable at this point is
$q^\ddagger\xrightarrow{-\vec e_{d+1}}q_r$, which yields a node $n$
labelled $(q_r,\vec v)$ for some $\vec v$ in $\+N^d$.  The deduction
subtree rooted by $n$ only uses rules of $\?B^\dagger$ and expansions,
thus as in the previous case $\?B^\dagger,q_\ell\jdg q_r,\vec v$.
Using increases in $q_r$ then shows $\?B^\dagger,q_\ell\jdg q_r,\vec
0$.  Therefore, in all cases:
\begin{claim}\label{cl-compr2}
  If $\?B^\ddagger,T^\ddagger_u\cup T_s,q_\ell\jdg_e
  q_r^\ddagger,\vec 0$, then $\?B^\dagger,q_\ell\jdg q_r,\vec 0$.
\end{claim}
By Claims~\ref{cl-cov-inc}, \ref{cl-compr1}, and~\ref{cl-compr2},
$\?B^\ddagger,T^\ddagger_u\cup T_s,q_\ell\jdg_e q_r^\ddagger,\vec 0$
if and only if there exists $\vec v$ in $\+N^d$ such that
$\?B,q_\ell\jdg q_r,\vec v$, thereby showing the correctness of our
reduction.

\subsection{Focusing \detitle{Proofs} in $\Rel$}\label{sub-focus}
We enforce a particular proof policy in our simulation of BVASSs in
$\Rel$, which is inspired by the \emph{focusing proof}
techniques~\citep{andreoli92} employed to reduce non-determinism
during proof search in sequent calculi.  With only implication at our
disposal, we find ourselves in a `negative fragment', where focusing
proofs have a very simple calculus $F\Rel$.  This is equivalent to
restricting oneself to long normal forms in the associated
$\lambda$-calculus.

A \emph{focusing sequent} is of one of the two forms
`$\Gamma,[A]\fprv B$' or `$\Gamma\fprv A$' where `$[A]$' is called a
\emph{focused} formula.  We let $\Gamma,\Delta,\dots$ denote as before
multisets of implicational formul\ae\ and $A,B,C$ implicational formul\ae.
Here are the rules of the focusing calculus $F\Rel$:
\begin{gather*}
    \tabrulelabelr[atomic identity]{fl-al}{}{[a]\fprv a}{atomic}\qquad
    \tabrulelabelr[focus]{fl-f}{\Gamma,[A]\fprv a}
    {\Gamma,A\fprv a}{focus}\qquad
    \tabrulelabelr[contraction]{fl-cn}{\Gamma, A, A\fprv a}{\Gamma,
    A\fprv a}{C^\mathit{f}}\\[1em]%
  \tabrulelabelr[left implication]{fl-impl}{\Gamma\fprv A\quad\Delta, [B]\fprv a}{\Gamma,\Delta,[A\imp B]\fprv a}{\imp_L^\mathit{f}}
  \qquad
  \tabrulelabelr[right implication]{fl-impr}{\Gamma, A\fprv B}{\Gamma\fprv A\imp B}{\imp_R^\mathit{f}}
\end{gather*}
Note that our focusing calculus $F\Rel$ gives the priority to right
implications \eqref{fl-impr} over the left
implications~\eqref{fl-impl}, focus~\eqref{fl-f}, and
contractions~\eqref{fl-cn}: the latter can only be applied to
sequents with atomic succedents $a$ in $\?A$.  A similar observation
is that a focusing sequent $\Gamma,[a]\fprv A$ is provable if and
only if $A=a$ is atomic and $\Gamma=\emptyset$ is the empty multiset,
since \eqref{fl-al} is the only rule yielding a sequent with a focused
atomic formula $[a]$.

\begin{restatable}[$F\Rel$ is sound and complete]{theorem}{thfocus}\label{th-focus}
  A sequent $\Gamma\vdash A$ is provable in $L\Rel$ if and only if the
  focusing sequent $\Gamma\fprv A$ is provable in $F\Rel$.
\end{restatable}\noindent
We prove \autoref{th-focus} in \appref{sec-focus}, using the
admissibility of a suitable cut rule in $F\Rel$.

\subsection{From \detitle{Comprehensive Expansive Reachability} to
  {\boldmath $F\Rel$}}\label{sub-hard}
Let us consider a comprehensive expansive reachability instance
$\tup{\?B,q_r,q_\ell}$ where $\?B=\tup{Q,d,T_u,T_s}$.  We are going to
construct an implicational formula $F$ such that $\vdash F$ if and
only if $\?B,T_u\cup T_s,q_\ell\jdg_e q_r,\vec 0$.

We work for this on the set of atomic formul\ae\ $Q\uplus\{e_i\mid
0<i\leq d\}$, and associate to a root judgement $\?B,T,q_\ell\jdg_e
q,\vec v$ a focusing sequent 
\begin{equation}
  q_\ell,\Delta_T,\Gamma_{\vec v}\fprv q
\end{equation}
where $\Delta_T$ encodes the rules in $T\subseteq T_u\cup T_s$ and
$\Gamma_{\vec v}$ encodes $\vec v$: let $T=\{t_1,\dots,t_k\}$ and
$\vec v=c_1\vec e_1+\cdots+c_d\vec e_d$, then
\begin{align}
  \Delta_T&\eqdef \enc{t_1},\dots,\enc{t_k}\;,\\
  \Gamma_{\vec v}&\eqdef e_1^{c_1},\dots,e_d^{c_d}\;,
  \shortintertext{%
    where `$\enc{t}$' is the individual encoding of rule $t$ and
    `$A^c$' stands for $c$ repetitions of the formula $A$.  We use the
    following individual rule encodings:}
  \enc{q\xrightarrow{\vec e_i}q_1}&\eqdef (e_i\imp q_1)\imp q\;,\\
  \enc{q\xrightarrow{-\vec e_i}q_1}&\eqdef q_1\imp(e_i\imp q)\;,\\
  \enc{q\to q_1+q_2}&\eqdef q_1\imp(q_2\imp q)\;.
\end{align}
Then proof search in $F\Rel$ is easily seen to implement deductions in
$\?B$:
\begin{claim}[Completeness]\label{cl-exp-fl1}
  If $\?B,T,q_\ell\jdg_e q,\vec v$, then
  $q_\ell,\Delta_T,\Gamma_{\vec v}\fprv q$.
\end{claim}
\begin{proof}
  We proceed by induction on the structure of the root judgement.  For
  the base case, i.e.\ for $\?B,\emptyset,q_\ell\jdg_e q_\ell,\vec
  0$, we have
  \begin{prooftree}
    \AxiomC{}
    \rlabel{fl-al}
    \UnaryInfC{$[q_\ell]\fprv q_\ell$}
    \rlabel{fl-f}
    \UnaryInfC{$q_\ell\fprv q_\ell$}
  \end{prooftree} as desired.

  For the induction step, if the last applied rule is an increment
  $t=q\xrightarrow{\vec e_i}q_1$ on a judgement $\?B,T,q_\ell\jdg_e
  q_1,\vec v+\vec e_i$, then
  \begin{prooftree}
    \AxiomC{i.h.}\noLine
    \UnaryInfC{$q_\ell,\Delta_T,\Gamma_{\vec v},e_i\fprv q_1$}
    \rlabel{fl-impr}
    \UnaryInfC{$q_\ell,\Delta_T,\Gamma_{\vec v}\fprv e_i\imp q_1$}
      \AxiomC{}
      \rlabel{fl-al}
      \UnaryInfC{$[q]\fprv q$}
    \rlabel{fl-impl}
    \BinaryInfC{$q_\ell,\Delta_T,\Gamma_{\vec v},[(e_1\imp q_1)\imp q]\fprv q$}
    \rlabel{fl-f}
    \UnaryInfC{$q_\ell,\Delta_T,\Gamma_{\vec v},(e_1\imp q_1)\imp q\fprv q$}
  \end{prooftree}
  and an additional contraction~\eqref{fl-cn} if $t\in T$
  shows $q_\ell,\Delta_{T\cup\{t\}},\Gamma_{\vec v}\fprv q$ as desired.

  If the last applied rule is a decrement $t=q\xrightarrow{-\vec
    e_i}q_1$ on a judgement $\?B,T,q_\ell\jdg_e q_1,\vec v$, then
  \begin{prooftree}
    \AxiomC{i.h.}\noLine
    \UnaryInfC{$q_\ell,\Delta_T,\Gamma_{\vec v}\fprv q_1$}
      \AxiomC{}
      \rlabel{fl-al}
      \UnaryInfC{$[e_i]\fprv e_i$}
      \rlabel{fl-f}
      \UnaryInfC{$e_i\fprv e_i$}
        \AxiomC{}
        \rlabel{fl-al}
        \UnaryInfC{$[q]\fprv q$}
      \rlabel{fl-impl}
      \BinaryInfC{$e_i,[e_i\imp q]\fprv q$}
    \rlabel{fl-impl}
    \BinaryInfC{$q_\ell,\Delta_T,\Gamma_{\vec v},e_i,[q_1\imp (e_i\imp q)]\fprv q$}
    \rlabel{fl-f}
    \UnaryInfC{$q_\ell,\Delta_T,\Gamma_{\vec v},e_i,q_1\imp (e_i\imp q)\fprv q$}
  \end{prooftree}
  and a contraction~\eqref{fl-cn} if $t\in T$
  shows $q_\ell,\Delta_{T\cup\{t\}},\Gamma_{\vec v+\vec e_i}\fprv q$ as desired.
  
  If the last applied rule is an expansion on a judgement
  $\?B,T,q_\ell\jdg_e q,\vec v+2\vec e_i$, then
  \begin{prooftree}
    \AxiomC{i.h.}\noLine
    \UnaryInfC{$q_\ell,\Delta_T,\Gamma_{\vec v},e_i,e_i\fprv q$}
    \rlabel{fl-cn}
    \UnaryInfC{$q_\ell,\Delta_T,\Gamma_{\vec v},e_i\fprv q$}
  \end{prooftree} as desired.

  Finally, if the last applied rule is a split $t=q\to q_1+q_2$ on two
  judgements $\?B,T_1,q_\ell\jdg_e q_1,\vec v_1$ and
  $\?B,T_2,q_\ell\jdg_e q_2,\vec v_2$, then
  \begin{prooftree}
    \AxiomC{i.h.}\noLine
    \UnaryInfC{$q_\ell,\Delta_{T_1},\Gamma_{\vec v_1}\fprv q_1$}
      \AxiomC{i.h.}\noLine
      \UnaryInfC{$q_\ell,\Delta_{T_2},\Gamma_{\vec v_2}\fprv q_2$}
        \AxiomC{}
        \rlabel{fl-al}
        \UnaryInfC{$[q]\fprv q$}
      \rlabel{fl-impl}
      \BinaryInfC{$q_\ell,\Delta_{T_2},\Gamma_{\vec v_2},[q_2\imp q]\fprv q$}
    \rlabel{fl-impl}
    \BinaryInfC{$q_\ell,q_\ell,\Delta_{T_1},\Delta_{T_2},\Gamma_{\vec
        v_1},\Gamma_{\vec v_2},[q_1\imp(q_2\imp q)]\fprv q$}
    \rlabel{fl-cn}
    \doubleLine
    \UnaryInfC{$q_\ell,\Delta_{T_1\cup T_2\cup \{t\}},\Gamma_{\vec
        v_1+\vec v_2}\fprv q$}
  \end{prooftree}
  as desired.
\end{proof}

The interest of the focusing calculus $F\Rel$ is that, starting from a
sequent $q_\ell,\Delta,\Gamma,[\enc{t}]\fprv q$ where $\Delta$ is in
$\+N^{\Delta_{T_u\cup T_s}}$ and $\Gamma$ in $\+N^{\{e_i\mid 0<i\leq
  d\}}$ and the focus is on the encoding of a rule $t$, there is
\emph{no choice} but to follow the proof trees shown in the proof of
\autoref{cl-exp-fl1}.  Given a multiset $m$ in $\+N^E$ for some set
$E$, we write
\begin{equation}
  \sigma(m)\eqdef\{e\in E\mid m(e)>0\}
\end{equation}
for the \emph{support} of $m$.  
\begin{restatable}[Soundness]{claim}{clexpfl}\label{cl-exp-fl2}
  Let $\Delta$ be in $\+N^{\Delta_{T_u\cup T_s}}$, $\Gamma$ in
  $\+N^{\{e_i\mid 0<i\leq d\}}$, $q$ in $Q$, and $n>0$.  Then
  $q_\ell^n,\Delta,\Gamma\fprv q$ implies
  $\?B,\sigma(\Delta),q_\ell\jdg_e q,\vec v_{\Gamma}$.
\end{restatable}
\begin{proof}Note that $n=0$ would yield an unprovable sequent.
  We proceed by induction on the structure of a proof tree for the
  focusing sequent.  The only applicable rules in a proof search from
  $q_\ell^n,\Delta,\Gamma\fprv q$ are~\eqref{fl-f}
  and~\eqref{fl-cn}.  In the latter case, we distinguish three cases
  depending on the contracted formula $A$:
  \begin{itemize}
  \item If $A=q_\ell$, i.e.\ if $q_\ell^{n+1},\Delta,\Gamma\fprv q$,
    then by induction hypothesis $\?B,\sigma(\Delta),q_\ell\jdg_e
    q,\vec v_{\Gamma}$ as desired.
  \item If $A=e_i$ in $\Gamma$, then by induction hypothesis
    $\?B,\sigma(\Delta),q_\ell\jdg_e q,\vec v_{\Gamma}+\vec e_i$,
    and an expansion yields $\?B,\sigma(\Delta),q_\ell\jdg_e q,\vec
    v_{\Gamma}$ as desired.
  \item If $A=\enc{t}$ for some rule $t$ in $\Delta$, then the support
    $\sigma(\Delta)$ is not changed and by induction hypothesis
    $\?B,\sigma(\Delta),q_\ell\jdg_e q,\vec v_{\Gamma}$ as desired.
  \end{itemize}

  In the former case, we also distinguish three cases depending on
  which formula $A$ receives the focus:
  \begin{itemize}
  \item If $A=q_\ell$, necessarily $\Delta$ and $\Gamma$ are
    empty and $q=q_\ell$, and indeed $\?B,\emptyset,q_\ell\jdg_e
    q_\ell,\vec 0$.
  \item If $A=e_i$ in $\Gamma$, then proof search fails since $q\neq
    e_i$. 
  \item If $A=\enc{t}$ in $\Delta$, then proof search needs to follow
    the proof trees used in the proof of \autoref{cl-exp-fl1}, and
    applying the induction hypothesis on the open leaves of these
    trees allows to conclude in each case%
    .\qedhere
  \end{itemize}
\end{proof}

\begin{theorem}\label{th-lowb}
  Provability in $\Rel$ is \textsc{2-ExpTime}-hard.
\end{theorem}
\begin{proof}
  We reduce from the comprehensive expansive reachability problem,
  which is \textsc{2-ExpTime}-hard by \autoref{prop-compr}.  From an
  instance $\tup{\?B,q_r,q_\ell}$ where $\?B=\tup{Q,d,T_u,T_s}$, we
  construct a formula $F\eqdef q_\ell\imp \varphi(T_u\cup T_s,q_r)$
  defined by
  \begin{align}
    \varphi(\emptyset,q_r)&\eqdef q_r\;,&
    \varphi(T\uplus\{t\},q_r)&\eqdef \enc{t}\imp \varphi(T,q_r)\;.
  \end{align}
  By~\autoref{th-focus}, $\vdash F$ if and only if $\fprv F$.  The
  latter holds if and only if $q_\ell,\Delta_{T_u\cup T_s}\fprv q_r$
  since we can only apply~\eqref{fl-impr}.  Then this occurs if and
  only if $\?B,T_u\cup T_s,q_\ell\jdg_e q_r,\vec 0$ by
  \autoref{cl-exp-fl1} and \autoref{cl-exp-fl2}.
\end{proof}

\section{Extensions}\label{sec-ext}
\ifshort
\else The complexity bounds proven in \autoref{th-upb} and
\autoref{th-lowb} for $\Rel$ can be generalised to larger fragments of
$\mathbf{R}$ and of propositional intuitionistic contractive linear
logic (\textbf{ILLC}).

\subsection{Adding \detitle{Multiplicatives}}
\fi
The sequent system $L\Rel$ for $\Rel$ can be extended to accommodate
further multiplicative connectives: the \emph{fusion} connective
$\circ$ (aka.\ `co-tenability' in \citep{entailment1}) and the
sentential constant $\mathbf t$:
\begin{equation*}
  \setlength{\arraycolsep}{8pt}
  \begin{array}{cc} 
    \tabrulelabelr[left truth]{rl-tl}{\Gamma\vdash A}{\Gamma,\mathbf
    t\vdash A}{\mathbf t_L}&
    \tabrulelabelr[right truth]{rl-tr}{}{\vdash\mathbf
    t}{\mathbf t_R}\\[1.5em]
    \tabrulelabelr[left fusion]{rl-fl}{\Gamma,A,B\vdash
    C}{\Gamma,A\circ B\vdash C}{\circ_L}&
    \tabrulelabelr[right fusion]{rl-fr}{\Gamma\vdash
    A\quad\Delta\vdash B}{\Gamma,\Delta\vdash A\circ B}{\circ_R}
  \end{array}
\end{equation*}
Let us call $L\mathbf{R}_{\imp,\circ}^{\mathbf t}$ the resulting
sequent system.  The BVASS $\?B_F$ presented in
\autoref{sub-rel-bvass} can be extended in a straightforward manner
with the rules of \autoref{fig-mll-bvass} and by identifying $q_\ell$
with $\mathbf t$.  Thanks to
\autoref{prop-exp-cov} and \autoref{cor-bvass}, this shows that
provability in $L\mathbf{R}_{\imp,\circ}^{\mathbf t}$ is in
\textsc{2-ExpTime}.
\begin{figure}[tbp]%
  \centering
  \begin{tikzpicture}[auto,node distance=.65cm]
    \node[state,label=left:{\ref{rl-tl}:}](Lone0){$A$};
    \node[state,right=of Lone0](Lone1){$A$};
    \path (Lone0) edge node{$-\vec e_{\mathbf t}$} (Lone1);
    \node[state,right=1.3cm of Lone1,label=left:{\ref{rl-fl}}](Lotimes0){$C$};
    \node[intermediate,right=1cm of Lotimes0](Lotimes1){};
    \node[state,right=1.2cm of Lotimes1](Lotimes2){$C$};
    \path (Lotimes0) edge node{$-\vec e_{A\circ B}$} (Lotimes1)
      (Lotimes1) edge node{$\vec e_A+\vec e_B$} (Lotimes2);
    \node[rstate,right=1.3cm of Lotimes2,label=left:{\ref{rl-fr}:},
      label=right:{$\,+$}](Rotimes0){$A\circ B$};
    \draw[draw=black!40] (Rotimes0.east) ++(.1,-.1) arc (-40:40:.14);
    \node[state,above right=0cm and .5cm of Rotimes0](Rotimes1){$A$};
    \node[state,below right=0cm and .5cm of Rotimes0](Rotimes2){$B$};
    \path (Rotimes0.east) edge (Rotimes1)
      (Rotimes0.east) edge (Rotimes2);
  \end{tikzpicture}\vspace*{-.5em}
  \caption{\label{fig-mll-bvass}The additional BVASS rules for $L\mathbf{R}_{\imp,\circ}^{\mathbf t}$.}
\end{figure}

Note that the sequent system $L\mathbf{R}_{\imp,\circ}^{\mathbf t}$ is
the same as that of intuitionistic multiplicative contractive linear
logic~(\textbf{IMLLC}), where $\imp$, $\circ$, and $\mathbf t$ are
usually noted respectively $\multimap$, $\otimes$, and $\mathbf 1$.

\ifshort
\else
\subsection{Adding \detitle{Exponentials}}
\label{sub-mellc}
\fi
In fact, essentially the same reduction from sequent calculus to
expansive BVASS reachability can be carried over for the more general
\emph{multiplicative exponential} fragment of intuitionistic
contractive linear logic with bottom $\mathbf{IMELZC}$, see
\citep[\propositionautorefname~9]{lazic14}.  The main differences are
that:
\begin{enumerate}
\item The exponential connectives incur an exponential blow-up in the
  number of states of the constructed BVASS $\?B_F$: its state space
  now contains $S\times 2^{S_!}$ where $S_!$ is the set of exponential
  subformul\ae\ of $F$.  The subsequent reduction to coverability
  in \autoref{sub-exp-cov} then performs a product with $2^{S\setminus
  S_!}$ (contractions in exponential subformul\ae\ being already
  handled), hence the resulting state space remains of size
  $2^{p(|F|)}$ for some polynomial $p$.
\item The exponential connectives also require an additional operation
  of \emph{full zero test}: as shown in
  \citep[\lemmaautorefname~3]{lazic14}, this operation can be
  eliminated at no cost in complexity.
\end{enumerate}
By~\autoref{th-lowb}, \autoref{prop-exp-cov}, and~\autoref{cor-bvass},
we conclude:
\begin{theorem}
  Provability in the logics $\Rel$, $\mathbf{R}^{\mathbf{t}}_{\imp}$, $\mathbf{IMLLC}$, 
$\mathbf{IMELLC}$, and $\mathbf{IMELZC}$
  is \textsc{2-ExpTime}-complete.
\end{theorem}
It seems likely that the non-intuitionistic variants $\mathbf{MLLC}$
and $\mathbf{MELLC}$, i.e.\ multiplicative and multiplicative
exponential contractive linear logic, are
also \textsc{2-ExpTime}-complete: the upper bound follows from the
bound on $\mathbf{IMELZC}$\ifshort\else, and the lower bound
should not cause any difficulty, but is beyond the scope of this
paper\fi.

\ifshort\else
\subsection{Adding \detitle{Additives}}
The key difference between $L\mathbf{R}_{\imp,\circ}^{\mathbf t}$ and
the \textsc{Ackermann}-complete sequent calculus studied by
\citet{urquhart99} for the conjunctive-implicational fragment
$\mathbf{R}_{\imp,\wedge}$ is the absence of conjunction $\wedge$ and
disjunction $\vee$, which are additive connectives that require a
richer model of \emph{alternating} BVASS; see~\citep{lazic14}.
\fi

\section{Concluding Remarks}\label{sec-concl}
Besides closing a longstanding open problem, the proof that $\Rel$
is \textsc{2-ExpTime}-complete paves the way for new investigations:
\begin{itemize}
\item In spite of the high worst-case complexity of BVASS
coverability, \citet{majumdar13} have recently presented a practical
algorithm with encouraging initial results.  The reduction
in \autoref{sec-upb} allows to transfer their
techniques to $\Rel$ provability, but might incur a
worst-case exponential blow-up.
\item Provability in the related implicational fragment $\mathbf
  T_\imp$ of ticket entailment has recently been proven decidable
  independently by \citet{padovani13} and \citet{bimbo13}.  Although
  the complexity of this problem is currently unknown, the latter
  proof relies on provability in $L\mathbf{R}^{\mathbf{t}}_{\imp}$,
  which we prove to be \textsc{2-ExpTime}-complete in
  \autoref{sec-ext}.
\end{itemize}

\subsection*{Acknowledgements}
The author thanks David Baelde for his excellent suggestion of
employing focusing proofs and helpful discussions around their uses.

\appendix
\section{Focusing Calculus}\label{sec-focus}
The purpose of this appendix is to prove \autoref{th-focus}: $F\Rel$
is sound and complete for $\Rel$.  We first prove some preliminary
statements:

\begin{lemma}[Right implication is invertible]\label{cl-invimpr}
  The sequent $\Gamma\fprv A\imp B$ is provable in $F\Rel$ if and
  only if $\Gamma,A\fprv B$ is provable.
\end{lemma}
\begin{proof}
  No other rule of $F\Rel$ can yield a non-atomic succedent $A\imp
  B$.
\end{proof}

\begin{lemma}[Identity is admissible]\label{cl-id}
  Let $A=A_n\imp\cdots\imp A_1\imp a$ where the $A_i$'s are arbitrary
  implicational formul\ae\ and $a$ is atomic.  Then the following rules
  are admissible in $F\Rel$:
  \begin{equation*}
    \tabrulelabelr[focused identity]{fl-idf}{}{A_1,\dots,A_n,[A]\fprv
      a}{Id_{[\,]}}\qquad
    \tabrulelabelr[identity]{fl-id}{}{A\fprv A}{Id^\mathit{f}}
  \end{equation*}
\end{lemma}
\begin{proof}
  Observe that \eqref{fl-id} can always be derived from \eqref{fl-idf}
  by
  \begin{prooftree}
    \AxiomC{}
    \rlabel{fl-idf}
    \UnaryInfC{$A_1,\dots,A_n,[A]\fprv a$}
    \rlabel{fl-f}
    \UnaryInfC{$A_1,\dots,A_n,A\fprv a$}
    \rlabel{fl-impr}
    \doubleLine
    \UnaryInfC{$A\fprv A$}
  \end{prooftree}
  
  We thus prove the two rules simultaneously by induction over the
  structure of $A$.  The base case for $A$ atomic (i.e.\ for
  $n=0$) holds thanks to the \eqref{fl-al} rule; for the induction
  step, we have the proof
  \begin{prooftree}
    \AxiomC{i.h.}
    \rlabel{fl-id}
    \UnaryInfC{$A_{i+1}\fprv A_{i+1}$}
      \AxiomC{i.h.}
      \rlabel{fl-idf}
      \UnaryInfC{$A_{i},\dots,A_1,[A_i\imp\cdots\imp A_1\imp a]\fprv a$}
    \rlabel{fl-impl}
    \BinaryInfC{$A_{i+1},A_i,\dots,A_1,[A_{i+1}\imp A_i\imp\cdots\imp A_1\imp a]\fprv a$}
  \end{prooftree}\vspace*{-1.4em}
\end{proof}

Let us introduce extra notation: lower-case $\gamma,\delta,\dots$
denote multisets of implicational or focused formul\ae\ containing
at most one focused formula.
\begin{lemma}[Mix is admissible]\label{cl-mix}
  The mix rules \eqref{fl-mix} and \eqref{fl-mixf} are admissible in
  $F\Rel$ for all $n\geq 0$:
  \begin{gather*}
    \tabrulelabelr[mix]{fl-mix}{\Gamma\fprv A\quad\delta,A^{n+1}\fprv
      a}{\Gamma,\delta\fprv a}{mix}\qquad
    \tabrulelabelr[mix]{fl-mixf}{\Gamma\fprv A\quad\Delta,A^{n},[A]\fprv
      a}{\Gamma,\Delta\fprv a}{mix_{[\,]}}
  \end{gather*}
\end{lemma}
\begin{proof}
  Assume that the two premises of a mix rule can be proved in $F\Rel$.
  We show by induction, first on the structure of the cut formula $A$,
  and second on the structure of the proof of the right premise, that
  the conclusion of the mix rule can be proved in $F\Rel$.  Let us
  consider for this the last rule applied in the proof of the right
  premise:
  \begin{description}
    \item[\nameref{fl-al}] this is only possible for
      \begin{prooftree}
        \AxiomC{$\Gamma\fprv a$}
          \AxiomC{}
          \rlabel{fl-al}
          \UnaryInfC{$[a]\fprv a$}
        \rlabel{fl-mixf}
        \BinaryInfC{$\Gamma\fprv a$}
      \end{prooftree}
      and could be obtained directly from the left premise.
    \item[\nameref{fl-f}] this is only possible for
      \begin{prooftree}
        \AxiomC{$\Gamma\fprv A$}
          \AxiomC{$\Delta,A^n,[A]\fprv a$}
          \rlabel{fl-f}
          \UnaryInfC{$\Delta,A^{n+1}\fprv a$}
        \rlabel{fl-mix}
        \BinaryInfC{$\Gamma,\Delta\fprv a$}
      \end{prooftree}
      and could be obtained directly from \eqref{fl-mixf} with a
      sub-proof of its right premise, thus by induction hypothesis there
      is a mix-free proof of $\Gamma,\Delta\fprv a$.
    \item[\nameref{fl-cn}] this is only possible for
      \begin{prooftree}
        \AxiomC{$\Gamma\fprv A$}
          \AxiomC{$\Delta,A^{n+2}\fprv a$}
          \rlabel{fl-cn}
          \UnaryInfC{$\Delta,A^{n+1}\fprv a$}
        \rlabel{fl-mix}
        \BinaryInfC{$\Gamma,\Delta\fprv a$}
      \end{prooftree}
      and could be obtained directly from \eqref{fl-mix} with a
      sub-proof of its right premise, thus by induction hypothesis
      there is a mix-free proof of $\Gamma,\Delta\fprv a$.
    \item[\nameref{fl-impl}] here we consider two cases depending on
      whether \eqref{fl-mix} or \eqref{fl-mixf} is applied.  Here is
      the situation with \eqref{fl-mix}:
      \begin{prooftree}
        \AxiomC{$\Gamma\fprv A$}
          \AxiomC{$\Delta_1,A^{n_1}\fprv B$}
            \AxiomC{$\Delta_2,A^{n_2},[C]\fprv a$}
          \rlabel{fl-impl}
          \BinaryInfC{$\Delta_1,\Delta_2,A^{n_1+n_2},[B\imp C]\fprv a$}
        \rlabel{fl-mix}
        \BinaryInfC{$\Gamma,\Delta_1,\Delta_2,[B\imp C]\fprv a$}
      \end{prooftree}
      We replace this proof by the following:
      \begin{prooftree}
        \AxiomC{$\Gamma\fprv A$}
          \AxiomC{$\Delta_1,A^{n_1}\fprv B$}
        \rlabel{fl-mix}
        \BinaryInfC{$\Gamma,\Delta_1\fprv B$}
            \AxiomC{$\Gamma\fprv A$}
              \AxiomC{$\Delta_2,A^{n_2},[C]\fprv a$}
            \rlabel{fl-mixf}
            \BinaryInfC{$\Gamma,\Delta_2,[C]\fprv a$}
        \rlabel{fl-impl}
        \BinaryInfC{$\Gamma,\Gamma,\Delta_1,\Delta_2,[B\imp C]\fprv
          a$}
        \rlabel{fl-cn}
        \doubleLine
        \UnaryInfC{$\Gamma,\Delta_1,\Delta_2,[B\imp C]\fprv a$}
      \end{prooftree}
      where the two introduced mix rules have sub-proofs of the
      original right premise as right premises, and can therefore be
      turned into mix-free proofs by induction hypothesis.

      The situation with \eqref{fl-mixf} is more involved:
      \begin{prooftree}
        \AxiomC{$\Gamma\fprv A\imp B$}
          \AxiomC{$\Delta_1,(A\imp B)^{n_1}\fprv A$}
            \AxiomC{$\Delta_2,(A\imp B)^{n_2},[B]\fprv a$}
          \rlabel{fl-impl}
          \BinaryInfC{$\Delta_1,\Delta_2,(A\imp B)^{n_1+n_2},[A\imp B]\fprv a$}
        \rlabel{fl-mixf}
        \BinaryInfC{$\Gamma,\Delta_1,\Delta_2\fprv a$}
      \end{prooftree}
      We prove this in two steps: we first prove
      $\Gamma,\Delta_2,A\fprv a$ by
      \begin{prooftree}
        \AxiomC{inv.\ \eqref{fl-impr}}
        \noLine
        \UnaryInfC{$\Gamma,A\fprv B$}
          \AxiomC{$\Gamma\fprv A\imp B$}
            \AxiomC{$\Delta_2,(A\imp B)^{n_2},[B]\fprv a$}
          \rlabel{fl-mix}
          \BinaryInfC{$\Gamma,\Delta_2,[B]\fprv a$}
        \rlabel{fl-mixf}
        \BinaryInfC{$\Gamma,A,\Gamma,\Delta_2\fprv a$}
        \rlabel{fl-cn}
        \doubleLine
        \UnaryInfC{$\Gamma,\Delta_2,A\fprv a$}
      \end{prooftree}
      This relies on \autoref{cl-invimpr} to show $\Gamma,A\fprv B$.
      By induction hypothesis, the top \eqref{fl-mix} is applied
      to a sub-proof of the right premise and the bottom
      \eqref{fl-mixf} uses a subformula $B$ of the original cut
      formula $A\imp B$, thus both have mix-free proofs.

      As a second step, we use $\Gamma,\Delta_2,A\fprv a$ in
      \begin{prooftree}
        \AxiomC{$\Gamma\fprv A\imp B$}
          \AxiomC{$\Delta_1,(A\imp B)^{n_1}\fprv A$}
        \rlabel{fl-mix}
        \BinaryInfC{$\Gamma,\Delta_1\fprv A$}
          \AxiomC{$\Gamma,\Delta_2,A\fprv a$}
        \rlabel{fl-mix}
        \BinaryInfC{$\Gamma,\Gamma,\Delta_1,\Delta_2\fprv a$}
        \rlabel{fl-cn}
        \doubleLine
        \UnaryInfC{$\Gamma,\Delta_1,\Delta_2\fprv a$}
      \end{prooftree}
      The top \eqref{fl-mix} has a sub-proof as right premise, while
      the bottom \eqref{fl-mix} uses a subformula $A$ of the cut
      formula $A\imp B$, thus both have mix-free proofs, yielding an
      overall mix-free proof for $\Gamma,\Delta_1,\Delta_2\fprv a$.
    \item[\nameref{fl-impr}] cannot yield a sequent with an atomic
      succedent $a$.\qedhere
  \end{description}
\end{proof}

\thfocus*
\begin{proof}
  First observe that any proof in $F\Rel$ can be transformed into a
  proof in $L\Rel$ by removing the focus information in sequents and
  suppressing \eqref{fl-f} rules.

  Conversely, we show by induction over the structure of the proof of
  $\Gamma\vdash A$ in $L\Rel$ that $\Gamma\fprv A$ in $F\Rel$.  To
  this end, let us examine the last applied rule in the proof of
  $\Gamma\vdash A$:
  \begin{description}
  \item[\nameref{rl-ax}] by \autoref{cl-id}, the corresponding
    sequent $A\fprv A$ can be proved in $F\Rel$.
  \item[\nameref{rl-cn}] let $B=B_1\imp\cdots\imp B_m\imp b$ where the
    $B_i$'s are implicational formul\ae\ and $b$ is an atomic formula.
    By induction hypothesis, $\Gamma,A,A\fprv B$ is provable, thus by
    \autoref{cl-invimpr}, we have the proof
    \begin{prooftree}
      \AxiomC{$\Gamma,A,A,B_1,\dots,B_m\fprv b$}
      \rlabel{fl-cn}
      \UnaryInfC{$\Gamma,A,B_1,\dots,B_m\fprv b$}
      \rlabel{fl-impr}
      \doubleLine
      \UnaryInfC{$\Gamma,A\fprv B$}
    \end{prooftree}
  \item[\nameref{fl-impl}] by induction hypothesis, $\Gamma\fprv A$
    and $\Delta,B\fprv C$.  Let $B=B_1\imp\cdots\imp B_m\imp b$ and
    $C=C_1\imp\cdots\imp C_n\imp c$ where $b$ and $c$ are atomic.
    Using \autoref{cl-id} and \autoref{cl-mix}, we first prove
  \begin{prooftree}
    \AxiomC{{\small i.h.}}
    \noLine
    \UnaryInfC{$\Gamma\fprv A$}
      \AxiomC{}
      \rlabel{fl-id}
      \UnaryInfC{$A\fprv A$}
        \AxiomC{}
        \rlabel{fl-idf}
        \UnaryInfC{$B_1,\dots,B_m,[B]\fprv b$}
      \rlabel{fl-impl}
      \BinaryInfC{$A,B_1,\dots,B_m,[A\imp B]\fprv b$}
      \rlabel{fl-f}
      \UnaryInfC{$A,B_1,\dots,B_m,A\imp B\fprv b$}
    \rlabel{fl-mix}
    \BinaryInfC{$\Gamma,B_1,\dots,B_m,A\imp B\fprv b$}
    \rlabel{fl-impr}
    \doubleLine
    \UnaryInfC{$\Gamma,A\imp B\fprv B$}
  \end{prooftree}
  Thus, using \autoref{cl-invimpr} and \autoref{cl-mix}, we have:
  \begin{prooftree}
    \AxiomC{$\Gamma,A\imp B\fprv B$}
          \AxiomC{{\small i.h.\ and inv.\ \eqref{fl-impr}}}
          \noLine
          \UnaryInfC{$\Delta,B,C_1,\dots,C_n\fprv c$}
    \rlabel{fl-mix}
    \BinaryInfC{$\Gamma,\Delta,A\imp B,C_1,\dots,C_n\fprv c$}
    \rlabel{fl-impr}
    \doubleLine
    \UnaryInfC{$\Gamma,\Delta,A\imp B\fprv C$}
  \end{prooftree}
  \item[\nameref{rl-impr}] the $F\Rel$ rule \eqref{fl-impr} is identical.\qedhere
  \end{description}
\end{proof}

\section{Branching Vector Addition Systems}\label{sec-bvas}
Our goal in this appendix is to prove \autoref{cor-bvass}.  Let us
first define branching VASs in a more formal way.  A BVAS is a tuple
$\?B=\tup{d,T_u,T_s}$ where both $T_u$ and $T_s$ are finite sets of
rules in $\+Z^d$.  The configurations of $\?B$ are thus in $\+N^d$,
and deductions follow the rules:
\begin{equation*}
  \tabrulelabelr[unary]{b-u}{\vec u+\vec v}{\vec v}{unary'}
  \qquad
  \tabrulelabelr[split]{b-s}{\vec u+\vec v_1+\vec v_2}{\vec v_1\quad \vec
    v_2}{split'}
\end{equation*}
for $\vec u$ in $T_u$, respectively $T_s$.%

The coverability problem then takes as input a BVAS $\?B$ of dimension
$d$, a root vector $\vec v_r$, and a leaf vector $\vec v_\ell$, both
in $\+N^d$, and asks for the existence of a deduction tree
using \eqref{b-u} and \eqref{b-s} with root labelled by some $\vec
v\geq\vec v_r$ for the product ordering over $\+N^d$ and leaves
labelled by $\vec v_\ell$.  Defining root judgements $\?B,\vec
v_\ell\jdg \vec v$ in the natural way, coverability then asks
whether there exists $\vec v\geq\vec v_r$ such that $\?B,\vec
v_\ell\jdg \vec v$.

\subsection{From BVAS to BVASS}
As mentioned earlier, BVAS coverability is easy to reduce to BVASS
coverability.  Given a BVAS coverability instance $\tup{\?B,\vec
v_r,\vec v_\ell}$ with $\?B=\tup{d,T_u,T_s}$, we build a BVASS
coverability instance $\tup{\?B',q_r,q_\ell}$ where
$\?B'\eqdef\tup{Q,d,T_u',T_s'}$ is equipped with a dummy state $q$,
one state $q_{\vec u}$ for each split rule $\vec u$ in $T_s$, and two
states $q_r$ and $q_\ell$.  Unary rules $\vec u$ in $T_u$ are encoded
as $q\xrightarrow{-\vec u}q$.  Split rules $\vec u$ in $T_s$ are
encoded as $q\xrightarrow{-\vec u}q_{\vec u}\to q+q$.  The root
condition is checked by a rule $q_r\xrightarrow{\vec v_r}q$, the leaf
condition by a rule $q\xrightarrow{-\vec v_\ell}q_\ell$.  Then a
deduction tree with root $\vec v$ and leaves $\vec v_\ell$ exists for
$\?B$ if and only if a deduction tree with root $(q,\vec v)$ and
leaves $(q,\vec v_\ell)$ exists for $\?B'$, and the result follows.

\subsection{From BVASS to BVAS}
The converse reduction is more involved.  \Citeauthor{demri12} sketch
in~\citep{demri12} a reduction that incurs a linear increase in the
dimension, but we need here to be more economical.

Given a vector $\vec v$ in $\+N^d$, we write $\|\vec v\|$ for its
infinite norm $\max_{0<i\leq d}\vec v(i)$.  For a set of vectors $T$,
$\|T\|$ then denotes $\max_{\vec v\in T}\|\vec v\|$.  When considering
BVAS(S) coverability problems, we assume a binary encoding of vectors,
e.g.\ $\|T_u\|\in 2^{O(|\?B|)}$.

\corbvass*
\begin{proof}
  Let us consider an instance $\tup{\?B,q_r,q_\ell}$ of the BVASS
  coverability problem, where $\?B=\tup{Q,d,T_u,T_s}$.  We construct a
  BVAS coverability instance $\tup{\?B',\vec v_r,\vec v_\ell}$ where
  $\?B'\eqdef\tup{6+d,T'_u,T'_s}$.  We do not assume $\?B$ to be
  ordinary, i.e.\ we have the following general rule for root
  judgements derived from unary rules in
  $T_u$: \begin{equation*} \frac{\?B,T,q_\ell\jdg q_1,\vec v+\vec u}
  {\?B,T\cup\{q\xrightarrow{\vec u}q_1\},q_\ell\jdg q,\vec
  v} \end{equation*}

  The idea of the reduction is to encode each state from the finite
  set $Q=\{q_0,\dots,q_{|Q|-1}\}$ using two extra dimensions.  For a
  state $q_i$ with $0\leq i<|Q|$, we write $\enc{q_i}$ for the pair
  $(i,|Q|-i)$.  This encoding verifies:
  \begin{align}\label{eq-binq}
    \enc q&\leq\enc{q'}&\text{ iff }&&q&=q'\;.
  \end{align}
  In order to implement unary and split rules from $\?B$ in $\?B'$, we
  actually introduce six new coordinates.  For a state $q$ in $Q$ and
  $0\leq k\leq 2$, we write $\enc{q}^k$ for the concatenation
  $0^{2k}\cdot\enc{q}\cdot 0^{2(2-k)+d}$.  For a vector $\vec u$ in
  $\+Z^d$, we also write $\vec u$ for its concatenation with $0^{6}$
  on its left.

  Let us define $\?B'$:
  \begin{align}
    \label{eq-bbu}
    T'_u&\eqdef\{\enc{q}^0-\enc{q_1}^1-\vec u\mid q\xrightarrow{\vec u}q_1\}\;,\\
    \label{eq-bbt}
    &\:\cup\:\{\enc{q}^1-\enc{q}^0\mid q\in
    Q\}\cup\{\enc{q}^2-\enc{q}^0\mid q\in Q\}\;,\\
    \label{eq-bbs}
    T'_s&\eqdef\{\enc{q}^0-\enc{q_1}^1-\enc{q_2}^2\mid q\to q_1+q_2\}\;,\\
    \label{eq-bbr}
    \vec v_r&\eqdef \enc{q_r}^0\;,\\
    \label{eq-bbl}
    \vec v_\ell&\eqdef\enc{q_\ell}^0\;.
  \end{align}
  Observe that $\|T_u'\|$ is in $O(|Q|+\|T_u\|)$, while $\|T'_s\|$,
  $\|v_r\|$ and $\|v_\ell\|$ are in $O(|Q|)$.
  We prove the correction of this construction by a series of claims:
  \begin{claim}\label{cl-bb1}
    If $\?B,q_\ell\jdg q,\vec v$, then $\?B',\vec
    v_\ell\jdg\enc{q}^0+\vec v$.
  \end{claim}
  By induction on the structure of a proof for the judgement
  $\?B,q_\ell\jdg q,\vec v$.
  \begin{itemize}
  \item For the base case,
    i.e.\ $\?B,q_\ell\jdg q_\ell,\vec 0$, then $\?B',\vec
    v_\ell\jdg\vec v_\ell=\enc{q_\ell}^0+\vec 0$ by~\eqref{eq-bbl}.
  \item For the induction step, first assume $\?B,q_\ell\jdg q,\vec
    v$ as the result of a unary rule $q\xrightarrow{u}q_1$.  Then by
    induction hypothesis, $\?B',\vec v_\ell\jdg \enc{q_1}^0+\vec
    v+\vec u$.  Applying~\eqref{eq-bbt} yields $\?B',\vec
    v_\ell\jdg\enc{q_1}^1+\vec v+\vec u$, from which~\eqref{eq-bbu}
    shows $\?B',\vec v_\ell\jdg\enc{q}^0+\vec v$ as desired.
  \item Still for the induction step, assume the root judgement
    results from a split rule $q\to q_1+q_2$.  Then by induction
    hypothesis, $\?B',\vec v_\ell\jdg\enc{q_1}^0+\vec v_1$, from
    which~\eqref{eq-bbt} yields $\?B',\vec
    v_\ell\jdg\enc{q_1}^1+\vec v_1$, and also by induction
    hypothesis $\?B',\vec v_\ell\jdg\enc{q_2}^0+\vec v_2$, from
    which~\eqref{eq-bbt} yields $\?B',\vec
    v_\ell\jdg\enc{q_2}^2+\vec v_2$.  Applying~\eqref{eq-bbs} then
    shows $\?B',\vec v_\ell\jdg\enc{q}^0+\vec v_1+\vec v_2$ as
    desired.
  \end{itemize}

  \begin{claim}\label{cl-bb2}
    If $0\leq k\leq 2$ and $\?B',\vec v_\ell\jdg\enc{q}^k+\vec v$,
    then $\?B,q_\ell\jdg q,\vec v$.
  \end{claim}
  By induction on the structure of a proof for the judgement $\?B',\vec
  v_\ell\jdg\enc{q}^k+\vec v$.
  \begin{itemize}
  \item For the base case, $\?B',\vec v_\ell\jdg\enc{q_\ell}^0$
    by~\eqref{eq-bbl}, and indeed $\?B,q_\ell\jdg q_\ell,\vec 0$.
  \item For the induction step, if the judgement results
    from~\eqref{eq-bbu}, then $\?B',\vec v_\ell\jdg\enc{q_1}^1+\vec
    v+\vec u$ with $k=0$.  By induction hypothesis $\?B,q_\ell\jdg
    q_1,\vec v+\vec u$, from which $q\xrightarrow{\vec u}q_1$ yields
    the desired root label.
  \item If the judgement results from~\eqref{eq-bbt}, then $k>0$ and
    the judgement has $\?B',\vec v_\ell\jdg \enc{q}^0+\vec v$ as premise,
    and the induction hypothesis allows to conclude directly.
  \item If the judgement results from~\eqref{eq-bbs}, then $\?B',\vec
    v_\ell\jdg\enc{q_1}^1+\vec v_1$ and $\?B',\vec
    v_\ell\jdg\enc{q_2}^2+\vec v_2$.  By the induction hypothesis
    $\?B,q_\ell\jdg q_1,\vec v_1$ and $\?B,q_\ell\jdg q_2,\vec
    v_2$, and an application of $q\to q_1+q_2$ yields the desired
    result.
  \end{itemize}
  In order to conclude on the correction, observe by an easy induction
  that $\?B',\vec v_\ell\jdg\vec w$ implies $\vec w=\enc{q}^k+\vec
  v$ for some $0\leq k\leq 2$, $q$ in $Q$, and $\vec v$ in $\+N^d$.
  Thus $\vec w\geq\vec v_r$ if and only if $\vec
  w=\enc{q_r}^0+\vec v$ for some $\vec v$ in $\+N^d$, and by the
  previous claims and \eqref{eq-binq}, if and only if
  $\?B,q_\ell\jdg q_r,\vec v$.
  
  Finally, regarding complexity, \citet{demri12} define in the
  proof of their \theoremautorefname~8 $L\eqdef \|T'_u\cup
  T'_s\|+\|\vec v_r\|+2$, which is in $O(|Q|+\|T_u\|)$, $H\eqdef
  L^{(3d)!}$, and $B\eqdef H^2$, and show that, if $\?B',\vec
  v_\ell\jdg \vec w\geq\vec v_r$, then there is a proof of the
  judgement of height at most $H$ and using vector values truncable to
  $B$.  The existence of such a proof can be established with an
  alternating Turing machine working in space $\log B$.  Since $\log
  B$ is in $2^{O(d\cdot\log d\cdot\log\log(|Q|+\|T_u\|))}$, this
  yields the stated bound on the complexity.
\end{proof}

\bibliographystyle{abbrvnat}
\bibliography{journalsabbr,conferences,relevance}

\end{document}